%% file: cvxHull.tex
\documentclass{ifacconf}

\usepackage{graphicx}      
\usepackage{natbib}        

\usepackage{amsmath}%
\usepackage{amsfonts}%
\usepackage{amssymb}%

\usepackage{algorithm}
\usepackage{optidef}
\usepackage{enumerate} 
\usepackage{tikz}

\usepackage{url}

\usepackage{mathtools}
\mathtoolsset{showonlyrefs}

\usepackage{algorithmic}
\usepackage{algorithm}

\usepackage{tikz}
\usepackage{tikz-3dplot}
\usetikzlibrary{shapes,arrows}

\newtheorem{theorem}{Theorem}

\newtheorem{definition}{Definition}

\newtheorem{problem}{Problem}

\newenvironment{proof}[1][Proof]{\quad \textit{#1.} }{\ \rule{0.5em}{0.5em}}

\DeclareMathOperator*{\argmax}{arg\,max}

\DeclareMathOperator{\diag}{diag}

\begin{document}

\input{danielcommands.tex}

\begin{frontmatter}

\title{Exact Set-valued Estimation using Constrained Convex Generators for uncertain Linear Systems} 

\thanks[footnoteinfo]{This work was partially supported by the Portuguese Funda\c{c}\~{a}o para a Ci\^{e}ncia e a Tecnologia (FCT) through Institute for Systems and Robotics (ISR), under Laboratory for Robotics and Engineering Systems (LARSyS) project UIDB/50009/2020, through project PCIF/MPG/0156/2019 FirePuma and through COPELABS, University Lus\'{o}fona project UIDB/04111/2020.}

\author[First]{Daniel Silvestre\thanksref{footnoteinfo}} 

\address[First]{School of Science and Technology from the NOVA University of Lisbon (FCT/UNL), 2829-516 Caparica, Portugal, with COPELABS from the Lus\'ofona University, and also with the Institute for Systems and Robotics, Instituto Superior T\'ecnico, University of Lisbon. (e-mail: dsilvestre@isr.tecnico.ulisboa.pt)}

\begin{abstract}                
Set-valued state estimation when in the presence of uncertainties in the model have been addressed in the literature essentially following three main approaches: i) interval arithmetic of the uncertain dynamics with the estimates; ii) factorizing the uncertainty into matrices with unity rank; and, iii) performing the convex hull for the vertices of the uncertainty space. Approach i) and ii) introduce a lot of conservatism because both disregard the relationship of the parameters with the entries of the dynamics matrix. On the other hand, approach iii) has a large growth on the number of variables required to represent the set or is approximated losing its main advantage in comparison with i) and ii). In this paper, with the application of autonomous vehicles in GPS-denied areas that resort to beacon signals for localization, we develop an exact (meaning no added conservatism) and optimal (smallest growth in the number of variables) closed-form definition for the convex hull of Convex Constrained Generators (CCGs). This results in a more efficient method to represent the minimum volume convex set corresponding to the state estimation. Given that reductions methods are still lacking in the literature for CCGs, we employ an approximation using ray-shooting that is comparable in terms of accuracy with methods for Constrained Zonotopes as the ones implemented in CORA. Simulations illustrate the greater accuracy of CCGs with the proposed convex hull operation in comparison to Constrained Zonotopes. 
\end{abstract}


\end{frontmatter}

\section{Introduction}
In the current state-of-the-art literature in autonomous systems, vehicles can use sensor measurements and set-valued observers for self-localization. The generated sets containing the true state and can be used for collision avoidance by checking intersection with the sets describing obstacles (\cite{ribeiro:20,ribeiro:21}). Therefore, these methods benefit from accurate set representations as conservatism will lead to worse trajectories or even unfeasibility. Considering range and bearing sensors corrupted by noise typically requires an over-approximation either using intervals (\cite{jaulin:11}) or ellipsoids (\cite{thor:08}). The introduction of Constrained Convex Generators (CCGs) in (\cite{silvestre:CCG}) enables representing both polytopes and ellipsoids with closed-form expressions for all the necessary set operations.

In the literature, the estimation has been carried in the stochastic setup with different Kalman filters according to the assumptions. Single beacon range measurement was tackled in (\cite{batista:11}) by a transformation of the nonlinear dynamics to obtain a Linear Time Varying (LTV) which allows for a Kalman Filter. The nonlinear model can be directly used by an Extended Kalman Filter (\cite{stilwell:05,hu:07,lee:07}). The stochastic approach is not desirable when a guaranteed state estimation is needed as in the case of fault-tolerant control, Model Predictive approaches, or vehicle collision detection with obstacles. 

The more general setup of estimation when considering uncertain Linear Parameter-Varying (LPV) has been performed for polytopes such as in (\cite{silvestre:17}). When there are no uncertainties, proposals using intervals (\cite{combastel:14}), zonotopes (\cite{combastel:03}) and ellipsoids (\cite{chernousko:05}) suffer from approximations during the intersection phase. Additionally, one can use ellipsotopes (\cite{ellipsotopes}) and AH-polytopes (\cite{ah-polytopes}) but not with uncertainties since there are currently no proposal of explicit formulas for the convex hull operation. Using polytopes in hyperplane representation (\cite{silvestre:detectable}) or in Constrained Zonotopes (CZs) (\cite{scott:16}) are the most predominant approaches. We remark that there is the possibility to represent the uncertainties as disturbance signals of varying intensity like the method in (\cite{silvestre:SVE}). Resorting to the equivalent techniques for nonlinear systems in (\cite{bonnifait:08}), (\cite{camacho:05}), (\cite{agung:09}), (\cite{raffo:18}), (\cite{shamma:18}), respectively, will result in unnecessary conservatism since we are disregarding the very specific structure of an uncertain LPV.

The convex hull operation for CZs has been introduced in (\cite{koeln:22}) and later generalized for CCGs in (\cite{dsilvestre:CCGapproxCvxHull}) at the expenses of adding a large number of generator variables equal to $3(n_{1}+n_{2})+1$, where $n_{1}$ and $n_{2}$ represent the number of original variables in both sets. In this paper, we provide a closed-form description with $n_{1}+n_{2}+1$ variables (and also a linear growth in the number of constraints) for CCGs. Given that CZs and ellipsoids are a particular case of CCGs, this also entails that their convex hull can be written with $n_{1}+n_{2}+1$ generators in the CCG format. The main contributions can be highlighted as:
\begin{itemize}
	\item Introduction of a closed-form expression for the convex hull that is exact for CCGs with optimal number of variables;
	\item The proposed method removes the growth factor associated with the convex hull, meaning that fewer order reduction procedures are required to maintain a tractable representation of the set-valued estimates.
\end{itemize}

The remainder of the paper is organized as follows. Section \ref{sec:problem} formalizes the state estimation problem, highlighting the exponential growth of the auxiliary variables. We review in Section \ref{sec:CCG} the definition and main set operations for CCGs, while Section \ref{sec:estimationCCG} is dedicated to presenting the proposed convex hull algorithm. Simulations using a unicycle model for a land autonomous vehicle are provided in Section \ref{sec:simulations}. Conclusions and directions of future work are given in Section \ref{sec:discussion}.

\textit{Notation }: We let $\0_n$ denote the $n$-dimensional vector of zeros and $I_{n}$ the identity matrix of size $n$. The operator $\mathrm{diag}(v)$ creates a diagonal matrix with $v$ in the diagonal or extracts the diagonal if the argument is a matrix. The transpose of a vector $v$ is denoted by $v\tp$, while the Euclidean norm for vector $x$ is represented as $\|x\|_{2}:= \sqrt{x\tp x}$. On the other hand, $\|x\|_{\infty}:= \max_{i} |x_{i}|$. The cartesian product is denoted by $\times$, the Minkowski sum of two sets by $\oplus$ and the intersection after applying a matrix $R$ to the first set by $\cap_{R}$.

\section{Problem Statement} \label{sec:problem}
The problem of state estimation in uncertain LPVs can be cast as finding a set of possible values given the measurements, disturbance, noise and initial state bounds when the model is given by:
\begin{equation}
	\label{eq:ulpv}
	\begin{aligned}
		x(k+1) &= \Big(F(\rho(k)) + \sum_{\ell = 1}^{n_{\Delta}} \Delta_{\ell}(k) U_{\ell}\Big) x(k) + B(\rho(k)) u(k) \\
		& \quad+ L(\rho(k)) d(k) \\
		y(k) &= C(\rho(k)) x(k) + N(\rho(k)) w(k) 
	\end{aligned}
\end{equation}
where $x(k) \in \R^{n}$, $u(k) \in \R^{n_{u}}$, $d(k) \in \R^{n_{d}}$, $y(k) \in \R^{m}$ and $w(k) \in \R^{n_{w}}$ are the system state, input, disturbance signal, output and noise, respectively. The parameter $\rho(k)$ is the part of the parameters that can be measured at time $k$, which can be treated as in the case of LTVs. The main challenge appears from considering the $n_{\Delta}$ uncertainties denoted by $\Delta_{\ell}$ and the constant matrices $U_{\ell}$ that account for how the uncertainties affect the nominal dynamics matrix given by $F(\rho(k))$. To lighten the notation, we will consider $F_{k}:= F(\rho(k))$ and similarly for all the remaining matrices in \eqref{eq:ulpv}. Notice that we have to explicitly consider $\rho$ to account for nonlinearities that enter the model in a linear fashion as will happen with unicycle model used in Section \ref{sec:simulations}. Moreover, in order to have a well-posed problem, we assume that all unknown signals are bounded within a compact convex set denoted by the correspondent capital letter, i.e., $x(0) \in X(0)$, $d(k) \in D(k)$ and $w(k) \in W(k)$. Without loss of generality, we will assume that $\forall k, |\Delta_{\ell}(k)| \leq 1$.

The problem addressed in this paper is summarized as:
\begin{problem}
	\label{prob:stateEstimation}
	Given compact convex sets $X(0)$, $D(k)$ and $W(k)$ for all $k \geq 0$ and measurements $y(k)$, how to compute a set $X(k)$ such that it is guaranteed that $x(k) \in X(k)$, $\forall k \geq 0$.
\end{problem}
Notice that Problem \ref{prob:stateEstimation} is called \emph{state estimation} although a converse definition could be presented for the output of the system (this is of particular interest in sensitivity analysis (\cite{silvestre:sensitivity}) and system distinguishability (\cite{silvestre:distinguish})). 

\section{Constrained Convex Generators overview} \label{sec:CCG}
In this section, we first review the main set operations and introduce the novel expression for the convex hull of the union of two CCGs. Definition \ref{def:CCG} and Definition \ref{def:CCGoperations} provide a formal description of CCGs and the required operations.
\begin{definition}[\cite{silvestre:CCG}]
	\label{def:CCG}
	A Constrained Convex Generator (CCG) $\mathcal{Z} \subset \R^{n}$ is defined by the tuple $(G,c,A,b,\mathfrak{C})$ with $G \in \R^{n\times n_{g}}$, $c \in \R^{n}$, $A \in \R^{n_{c}\times n_{g}}$, $b \in \R^{n_{c}}$, and $\mathfrak{C} := \{\mathcal{C}_{1}, \mathcal{C}_{2}, \cdots, \mathcal{C}_{n_{p}}\}$ such that:
	\begin{equation}
		\mathcal{Z} = \{G \xi + c :  A\xi = b, \xi \in \mathcal{C}_{1} \times \cdots \times \mathcal{C}_{n_{p}}\}.
	\end{equation}
\end{definition}
\begin{definition}[\cite{silvestre:CCG}]
	\label{def:CCGoperations}
	Consider three CCGs as in Definition \ref{def:CCG}:
	\begin{itemize}
		\item $Z = (G_{z}, c_{z}, A_{z}, b_{z},\mathfrak{C}_{z}) \subset \R^{n}$;
		\item $W = (G_{w}, c_{w}, A_{w}, b_{w},\mathfrak{C}_{w}) \subset \R^{n}$;
		\item $Y = (G_{y}, c_{y}, A_{y}, b_{y},\mathfrak{C}_{y}) \subset \R^{m}$;
	\end{itemize}
	and a matrix $R \in \R^{m \times n}$ and a vector $t \in \R^{m}$. The three set operations are defined as:
	\begin{equation}
		\label{eq:CCGlinearmap}
		RZ+t = \left(RG_{z}, Rc_{z}+t, A_{z},b_{z},\mathfrak{C}_{z}\right)
	\end{equation}
	\begin{equation}
		\label{eq:CCGminkowski}
		\resizebox{\hsize}{!}{$
			Z \oplus W = \left(\begin{bmatrix}
				G_{z} & G_{w}
			\end{bmatrix}, c_{z} + c_{w}, \begin{bmatrix}
				A_{z} & \0 \\
				\0 & A_{w}
			\end{bmatrix},\begin{bmatrix}
				b_{z} \\ b_{w}
			\end{bmatrix},\{\mathfrak{C}_{z},\mathfrak{C}_{w}\}\right)
			$}
	\end{equation}
	\begin{equation}
		\label{eq:CCGintersection}
		\resizebox{\hsize}{!}{$
			Z \cap_{R} Y = \left(\begin{bmatrix}
				G_{z} & \0
			\end{bmatrix}, c_{z}, \begin{bmatrix}
				A_{z} & \0 \\
				\0 & A_{y} \\
				RG_{z} & -G_{y}
			\end{bmatrix},\begin{bmatrix}
				b_{z} \\ b_{y} \\ c_{y} - Rc_{z}
			\end{bmatrix},\{\mathfrak{C}_{z},\mathfrak{C}_{y}\}\right).$}
	\end{equation}
\end{definition}

We would like to point out that all the aforementioned set representations are subsets of CCGs, namely:
\begin{itemize}
	\item an interval corresponds to $(G,c,[\,],[\,],\|\xi\|_{\infty} \leq 1)$, for a diagonal matrix $G$;
	\item a zonotope is given by $(G,c,[\,],[\,],\|\xi\|_{\infty} \leq 1)$;
	\item an ellipsoid is defined by $(G,c,[\,],[\,],\|\xi\|_{2} \leq 1)$, for a square matrix $G$;
	\item a CZ or polytope is $(G,c,A,b,\|\xi\|_{\infty} \leq 1)$;
	\item a convex cone in $\R^{n}$ is $(G,c,[\,],[\,],\xi \geq 0)$;
	\item ellipsotopes are given by $(G,c,A,b,\|\xi\|_{p_{1}} \leq 1, \cdots$ \\ $, \|\xi\|_{p_{m}} \leq 1)$, for some $p_{i} > 0$, $1 \leq i \leq m$;
	\item AH-polytopes are given by $(G,c,[\,],[\,],A \xi \leq b)$.
\end{itemize}

\section{State Estimation for uncertain LPVs using Constrained Convex Generators (CCGs)} \label{sec:estimationCCG}

In this section, we first recover the standard procedure to carry state estimation in the case of an uncertain LPV and then introduce the improved representation directly in CCG format that improves the work in (\cite{dsilvestre:CCGapproxCvxHull}). The propagation phase using the model has to account for the polytopic description of the uncertainty space, namely, set $X_{\mathrm{prop}}(k+1)$ after applying the dynamics to $X(k)$ can be written as:
\begin{equation}
	\label{eq:propagate}
	\resizebox{\hsize}{!}{$
		\begin{aligned}
			X_{\mathrm{prop}}(k+1) = & \mathrm{cvxHull}\left(\bigcup_{\Delta \in \mathrm{vertex}([-1,1]^{n_{\Delta}})} \left(F_{k} + \sum_{\ell = 1}^{n_{\Delta}} \Delta_{\ell}(k) U_{\ell}\right) X(k)\right)\\
			& + B_{k} u(k) \oplus L_{k} D(k),
		\end{aligned}
		$}
\end{equation}
where $\mathrm{cvxHull}$ is the convex hull function.

The update phase corresponding to intersection with the measurement set $Y(k+1)$, i.e., all state values that could result in the measurement $y(k+1)$, that corresponds to:
\begin{equation}
	\label{eq:update}
	X(k+1) = X_{\mathrm{prop}}(k+1) \cap_{C} Y(k+1).
\end{equation}

\subsection{Convex Hull for CCGs}
Let us start by defining the convex hull of two sets:
\begin{equation}
	\label{eq:cvxHull}
	\begin{aligned}
		\mathrm{cvxHull}\left(Z_{1},Z_{2}\right) := \big\{z: & z = \lambda z_{1} + (1-\lambda) z_{2}, \\
		&\lambda \in [0,1], z_{1} \in Z_{1}, z_{2} \in Z_{2}\big\}.
	\end{aligned}	
\end{equation}

Let us introduce a specific instance of norm cones that are going to be used in the following result. For a norm unity ball $\mathfrak{C}$ defined as $\|\xi\|_{p} \leq 1$, let us associate with it the correspondent norm cone of order zero $\mathfrak{C}^{(0)}(\xi,\lambda,a,b) := \|\xi\|_{p} + w_{0} \lambda \leq v_{0}$ with the initialization of the row vector $w_{0}$ and column vector $\lambda$ as empty and scalar $v_{0} = 1$. In the base case, we can omit the arguments with a slight abuse of notation. We can now define norm cones of higher order of this operation in a recursive manner $\mathfrak{C}^{(\tau)}(\xi,\lambda,a,b) := \|\xi\|_{p} + \begin{bmatrix}
	a & bw_{\tau-1}
\end{bmatrix} \lambda \leq b v_{\tau-1}$, such that the generator variables are $\lambda \in \R^{\tau}$ and $\xi$ with the same dimension as the zero order cone and constant arguments $a$ and $b$. 

We can now state the main theorem introducing the closed-form expression for the convex hull of two CCGs and the complexity of this representation.

\begin{theorem}
	\label{pro:cvxHull}
	Consider two Constrained Convex Generators (CCGs) as in Definition \ref{def:CCG}:
	\begin{itemize}
		\item $X = (G_{x}, c_{x}, A_{x}, b_{x},\mathfrak{C}^{(\tau_{x})}_{x}) \subset \R^{n}$;
		\item $Y = (G_{y}, c_{y}, A_{y}, b_{y},\mathfrak{C}^{(\tau_{y})}_{y}) \subset \R^{n}$;
	\end{itemize}
	such that $A_{x} \in \R^{n_{c}^{x} \times n_{g}^{x}}$, $A_{y} \in \R^{n_{c}^{y} \times n_{g}^{y}}$, $\xi_{x} \in \mathfrak{C}^{(\tau_{x})}_{x} \implies \alpha \xi_{x} \in \mathfrak{C}^{(\tau_{x})}_{x}$, for $\alpha \in [0,1]$ and similarly for $\mathfrak{C}^{(\tau_{y})}_{y}$. The CCG corresponding to the exact convex hull $Z_{h} = (G_{h}, c_{h}, A_{h}, b_{h},\mathfrak{C}_{h}) \subset \R^{n}$ is given by:
	\begin{equation}
		\label{eq:Zh}
		\begin{aligned}
			G_{h} &= \begin{bmatrix}
				G_{x} & G_{y} & c_{x}-c_{y}
			\end{bmatrix}, c_{h} = \frac{c_{x}+c_{y}}{2},\\
			A_{h} &= \begin{bmatrix}
						A_{x} & \0 & -b_{x}\\
						\0 & A_{y} & b_{y} 
	                 \end{bmatrix}, b_{h} = \begin{bmatrix}
				\frac{1}{2} b_{x}\\
				\frac{1}{2} b_{y}
			\end{bmatrix} \\
			\mathfrak{C}_{h} &= \{\mathfrak{C}^{(\tau_{x}+1)}_{x}(\xi_{x},\xi_{\lambda},-1,0.5),\mathfrak{C}^{(\tau_{y}+1)}_{y}(\xi_{y},\xi_{\lambda},1,0.5),\R\},
		\end{aligned}
	\end{equation}
	which has $n_{g}^{x}+n_{g}^{y}+1$ generators and $n_{c}^{x}+n_{c}^{y}$ constraints.
\end{theorem}
\begin{proof}
	Following Theorem 1 from (\cite{optimalCvxHull}), we write $Z_{h}$ as:
	\begin{equation}
		\begin{aligned}
			Z_{h} = \{p_{h} &= G_{x} \xi_{x} + \lambda c_{x} + G_{y}\xi_{y}+(1-\lambda)c_{y}:\\
							& 0 \leq \lambda \leq 1, A_{x}\xi_{x} = \lambda b_{x}, A_{y}\xi_{y} = (1-\lambda) b_{y},\\
							& \|\xi_{x}\|_{\ell_{x}} \leq \lambda, \|\xi_{y}\|_{\ell_{y}} \leq (1-\lambda)\}
		\end{aligned}
	\end{equation}
	when in the presence of unit balls. 
	
	By performing the substitution $\xi_{\lambda} = \lambda-0.5$, we obtain a generator variable that belongs to the interval $[-0.5,0.5]$ and after reorganizing to write everything in terms of $\xi_{h} = \begin{bmatrix}
		\xi_{x}\tp & \xi_{y}\tp & \xi_{\lambda}\tp
	\end{bmatrix}\tp $, we obtain:
	\begin{equation}
		\begin{aligned}
			Z_{h} = \{p_{h} &= G_{h}\xi_{h} + c_{h}:\\
			& A_{h}\xi_{h} = b_{h}, \|\xi_{x}\|_{\ell_{x}} \leq 0.5+\xi_{\lambda}, \|\xi_{y}\|_{\ell_{y}} \leq 0.5-\xi_{\lambda}\}.
		\end{aligned}
	\end{equation}
	where the norm cones correspond to $\mathfrak{C}^{(1)}_{x}(\xi_{x},\xi_{\lambda},-1,0.5)$ and $\mathfrak{C}^{(1)}_{y}(\xi_{y},\xi_{\lambda},1,0.5)$. If on the other hand, we have a norm cones of order $\tau_{x}$ and $\tau_{y}$, respectively, we obtain the expression $\mathfrak{C}^{(\tau_{x}+1)}_{x}(\xi_{x},\xi_{\lambda},-1,0.5)$ and $\mathfrak{C}^{(\tau_{y}+1)}_{y}(\xi_{y},\xi_{\lambda},1,0.5)$. The number of generators and constraints results from the size of matrix $A_{h}$, which concludes the proof.
\end{proof}

Theorem \ref{pro:cvxHull} produces the exact convex hull as no approximation was required and is available in the toolbox \textrm{ReachTool} that can be found in \url{https://github.com/danielmsilvestre/ReachTool}. Figure \ref{fig:ccgHull} depicts an example of sets $Z_{1}$ and $Z_{2}$ with the respective set $Z_{h}$ as given by Theorem \ref{pro:cvxHull} and what one would get if first converted the sets to CZs and then applied the exact convex hull given in (\cite{koeln:22}). As observed, the proposed method is tight for CCGs and offers better accuracy in comparison to the result from (\cite{koeln:22}). Moreover, since CZs are an instance of CCGs where the generator sets are $\ell_{\infty}$ unit balls, a corollary from Theorem \ref{pro:cvxHull} is that the optimal representation of the convex hull of two CZs is possible only in the more general CCG format. 

\begin{figure}
	\centering
	\includegraphics[scale=0.47,trim=100 271 120 283,clip=true]{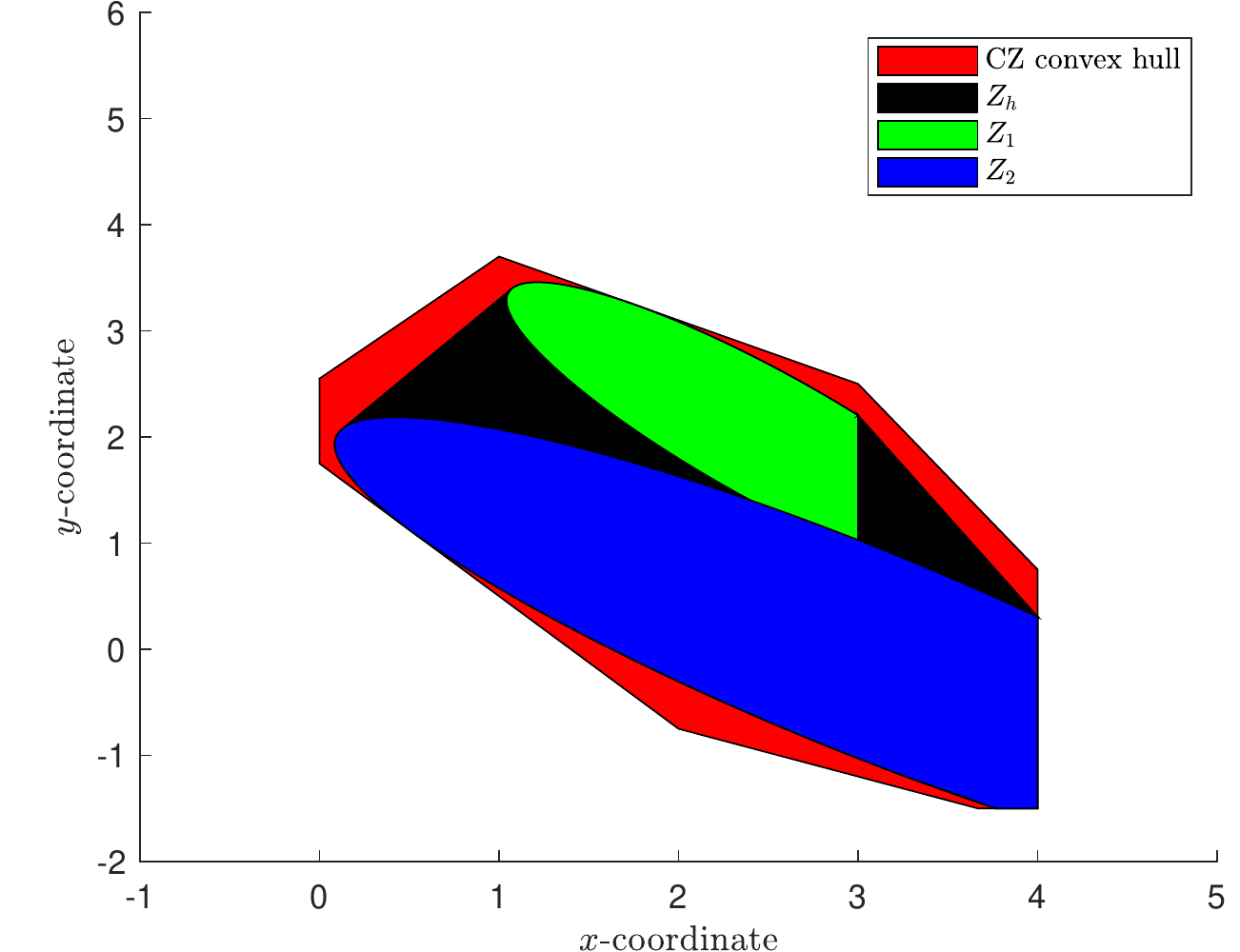}
	\caption{Comparison between the set $Z_{h}$ and the convex hull that one would obtain if first converted both $Z_{1}$ and $Z_{2}$ to constrained zonotopes by overbounding all convex generators by the $\ell_{\infty}$ unit ball.}
	\label{fig:ccgHull}
\end{figure}

The convex hull operator increases linearly the number of auxiliary variables to $n_{g}^{x} + n_{g}^{y} + 1$, however, this step has to be performed for all vertices which are exponential in the number of uncertainties. Such an issue was already present in (\cite{silvestre:17}) for polytopic set descriptions using the optimal convex hull formulation.

In order to keep the computation time for each iteration bounded, we introduce the order reduction in Algorithm \ref{alg:orderRedAlgorithm}, which computes a CCG with a specified number of constraints $\gamma$ using $n + \gamma$ generators which is of the form of a polytope. The procedure starts by constructing a collections of hyperplanes tangent to the surface and then converting to CCG representation. The $\min$ and $\max$ operations are element-wise.

\algsetup{indent=0.5em}

\begin{algorithm}
	\caption{Order Reduction using points on the surface.}\label{alg:orderRedAlgorithm}
	\begin{algorithmic}[1]
		\REQUIRE Set $X(K) \subseteq \R^{n}$ and desired order $\gamma$.
		\ENSURE Calculation of $X(k) \subseteq X_{\mathrm{red}}(k) \subseteq \R^{n}$ with $n_{g} = \gamma + n$ generators and $n_{c} = \gamma$ constraints.
		
		\medskip
		\STATE /* \textit{Get points $p_i$ on the surface such that $p_i = \argmax v_i\tp p_i$, $1 \leq i \leq \gamma$} for random $v_i$ */
		\STATE $[v,p] = \mathrm{sampleSurface}(X(k), \gamma)$ 
		\STATE /* \textit{Compute box $\tilde{Z}$ for the points $p$} */
		\STATE $\resizebox{0.93\hsize}{!}{$\tilde{Z} = (\frac{1}{2}\mathrm{diag}(\max p - \min p),\frac{1}{2}(\max p + \min p),[~],[~],\| \tilde{\xi}\|_{\infty} \leq 1)$}$
		\STATE /* \textit{Calculate $b$ and $\sigma$ such that all entries $v_i\tp p_i \in [\sigma,b]$}*/
		\STATE $\sigma = \min v\tp p$
		\STATE $b = \diag(v\tp p)$
		\STATE $\resizebox{0.92\hsize}{!}{$\begin{aligned}
				X_{\mathrm{red}}(k) = (&\begin{bmatrix}
					\tilde{Z}.G & \0_{n \times \gamma}
				\end{bmatrix},\tilde{Z}.c,\begin{bmatrix}
					v\tp \tilde{Z}.G &\frac{1}{2} \mathrm{diag}(\sigma - b)
				\end{bmatrix},\\
				&\frac{b+\sigma}{2}-v\tp \tilde{Z}.c,\| \tilde{\xi}\|_{\infty} \leq 1)
			\end{aligned}$}$ 
	\end{algorithmic}
\end{algorithm}

\section{Simulations} \label{sec:simulations}

In this section, simulations results are presented for a unicycle model of an autonomous vehicle in discrete-time for which there is a digital compass as an onboard sensor providing measurements of the orientation angle with a $\pm 5^{\circ}$ error. Simulations were run in Matlab R2018a running on a HP machine with a Intel Core i7-8550U CPU @ 1.80GHz and 12 GB of memory resorting to Yalmip as the language to model optimization problems and Mosek as the underlying solver. Videos, figures and code can be found in \url{https://github.com/danielmsilvestre/CCGExactConvexHull}

We recover the example considering unicycle dynamics described in (\cite{bricaire:11}). The vehicle schematic representation is given in Figure \ref{fig:unicycle} and has the following dynamics in discrete-time:
\begin{equation}
	\label{eq:dtUnicycle}
	\begin{bmatrix}
		p_{i}\\
		q_{i}
	\end{bmatrix} (k+1) = 
	\begin{bmatrix}
		p_{i}\\
		q_{i}
	\end{bmatrix} (k) + \mathrm{Ts} ~ A_{i}(\theta_{i})
	\begin{bmatrix}
		v_{i}\\
		w_{i}
	\end{bmatrix}(k)
\end{equation}
where the state $(p_{i},q_{i})$ identify the position of the front of the $i$th vehicle and the inputs $(v_{i},w_{i})$ account for the linear velocity and rotation. Moreover, $\mathrm{Ts} = 0.1$ stands for the sampling time, $\theta_{i}$ (we omit the time dependence in $k$ for a more compact presentation) for the orientation and matrix $A_{i}(\theta_{i})$ is given as:
\begin{equation}
	A_{i}(\theta_{i}) = \begin{bmatrix}
		\cos \theta_{i} & -l\sin\theta_{i}\\
		\sin\theta_{i} & l\cos \theta_{i}
	\end{bmatrix}.
\end{equation}
In this simulation, we consider a single vehicle running for a total of 15 seconds and, assuming that the compass takes measurements $\hat{\theta_{1}}$ of the true variable $\theta_{1}$ that have a maximum of $\pm 5^{\circ}$ following a uniform distribution. Therefore, at each iteration time $k$, matrix $A_{1}$ in the dynamics is not available to the observer and we have to consider $\hat{\theta_{1}}$ to generate the nominal dynamics and an uncertainty $\Delta_{1}$ with maximum magnitude of $5^{\circ}$, which fits \eqref{eq:ulpv}. 

\begin{figure}
	\centering
	\includegraphics[scale=0.18,trim=0 200 0 160]{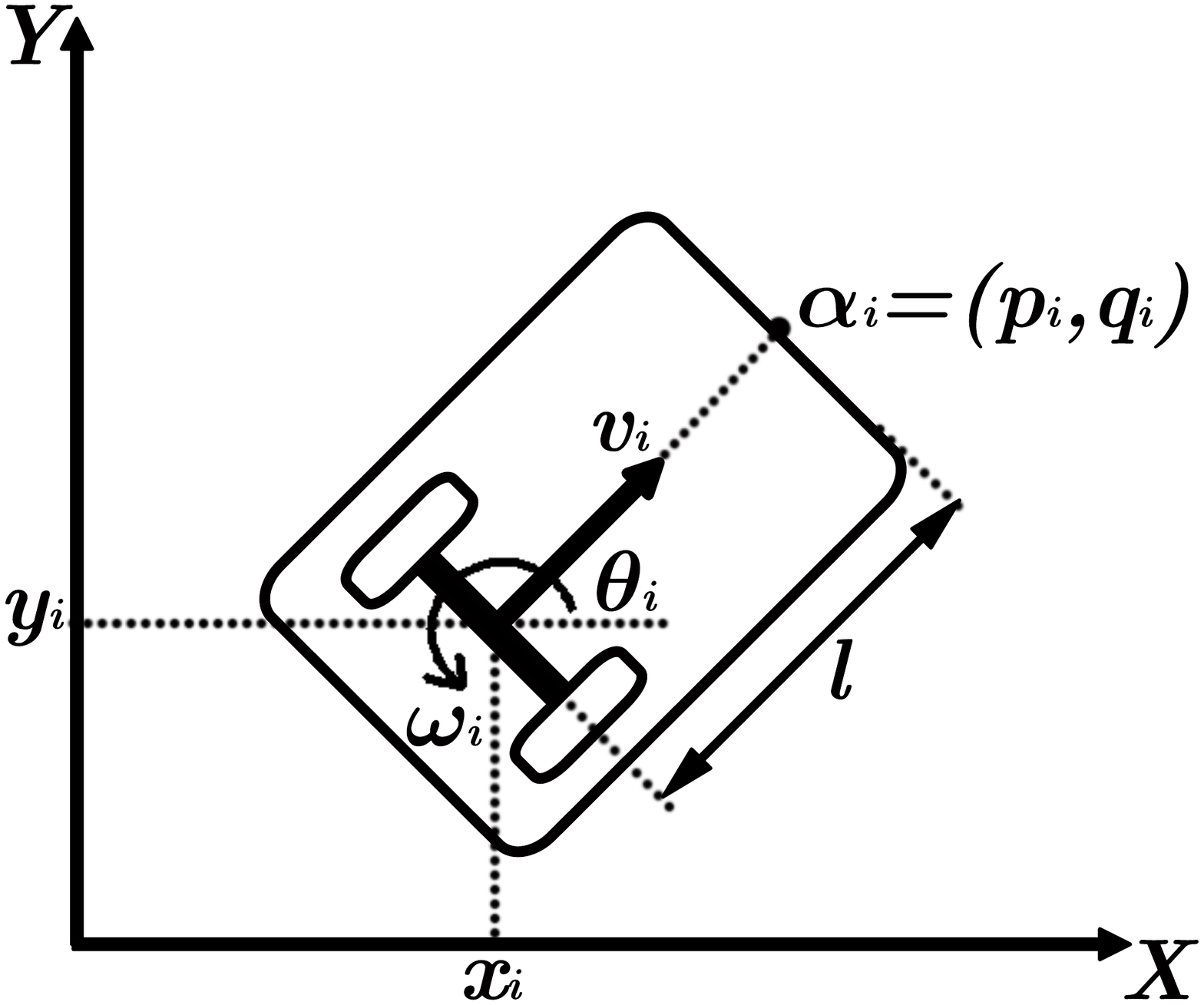}
	\caption{Schematic of the unicycle model for the vehicles.}
	\label{fig:unicycle}
\end{figure}

The trajectory-tracking control law used is:
\begin{equation}
	\label{eq:vehicleController}
	\resizebox{\hsize}{!}{$
		\left[\begin{array}{l}
			v_{i}(k) \\
			w_{i}(k)
		\end{array}\right]=\frac{A_{i}^{-1}\left(\theta_{i}\right)}{\mathrm{Ts}} \left(\tau(k+1) -\frac{\tau(k)}{2} - 0.5\left[\begin{array}{l}
			p_{i}(k) \\
			q_{i}(k)
		\end{array}\right] + d(k)\right)$}
\end{equation}
where $\tau(k)$ accounts for the discrete sequence of waypoints in the trajectory. Once again, we assume that there is a telemetry sensor that produces estimates corrupted by noise of the value of $p_{1}(k)$ and $q_{1}(k)$ and add the corresponding disturbance term $d(k)$ to account for those differences. Moreover, there are two beacons at positions $\begin{bmatrix}
	5 & 25
\end{bmatrix}\tp$ and $\begin{bmatrix}
	23 & 10
\end{bmatrix}\tp $ that can be detected within a 5 and 2 units of distance which allows to better localize the vehicle.

The vehicle performs a figure 8 trajectory such that it can only get measurements from each beacon in one time interval. Figure \ref{fig:fig1-volumes} illustrates the volume evolution for the set-valued estimates $X(k)$ when using CZs (\cite{scott:16}) and CCGs when both used the same order reduction method in Section \ref{sec:estimationCCG}. Since the vehicle is moving and most of the time performing dead reckoning with the uncertain LPV model, the volume keeps increasing and is lowered when the vehicle reaches the beacon areas. The main trend to observe is that the added accuracy of the $\ell_{2}$ ball representing the range measurement from the beacon contributes to a better performance of the CCG filter.

\begin{figure}
	\centering
	\includegraphics[scale=0.47,trim=100 271 120 274,clip=true]{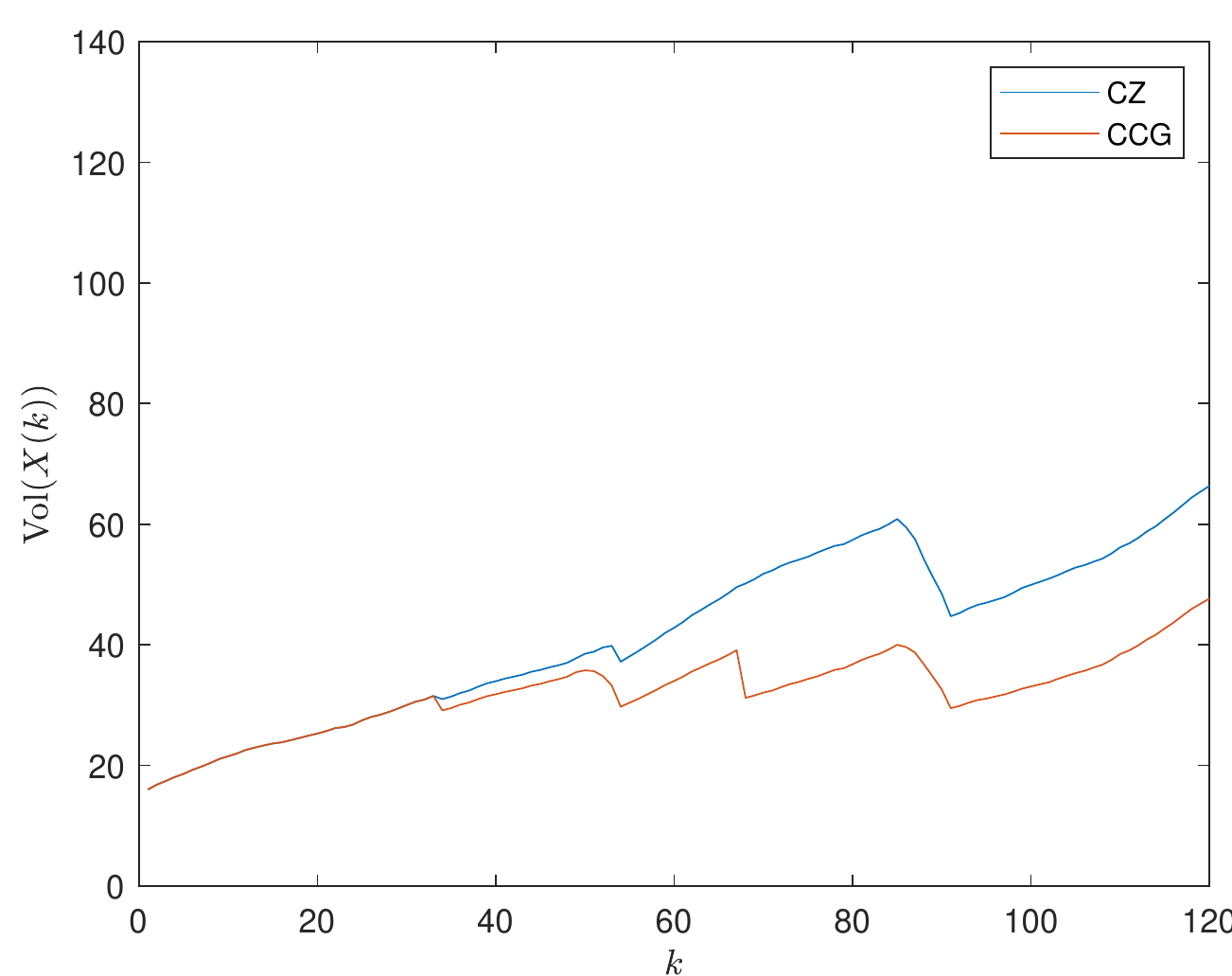}
	\caption{Comparison of the volume for both set-valued estimates when using constrained zonotopes (CZ) and CCGs for the figure 8 trajectory.}
	\label{fig:fig1-volumes}
\end{figure}

In Figure \ref{fig:fig2-F8}, it is illustrated the trajectory executed by the vehicle and the corresponding set-valued estimates using both the CZ and CCG approaches. We have selected a small number of time instants to display the sets as to avoid cluttering the image, but the full video can be found in the GitHub repository associated with the paper. 

\begin{figure}
	\centering
	\includegraphics[scale=0.47,trim=100 271 120 283,clip=true]{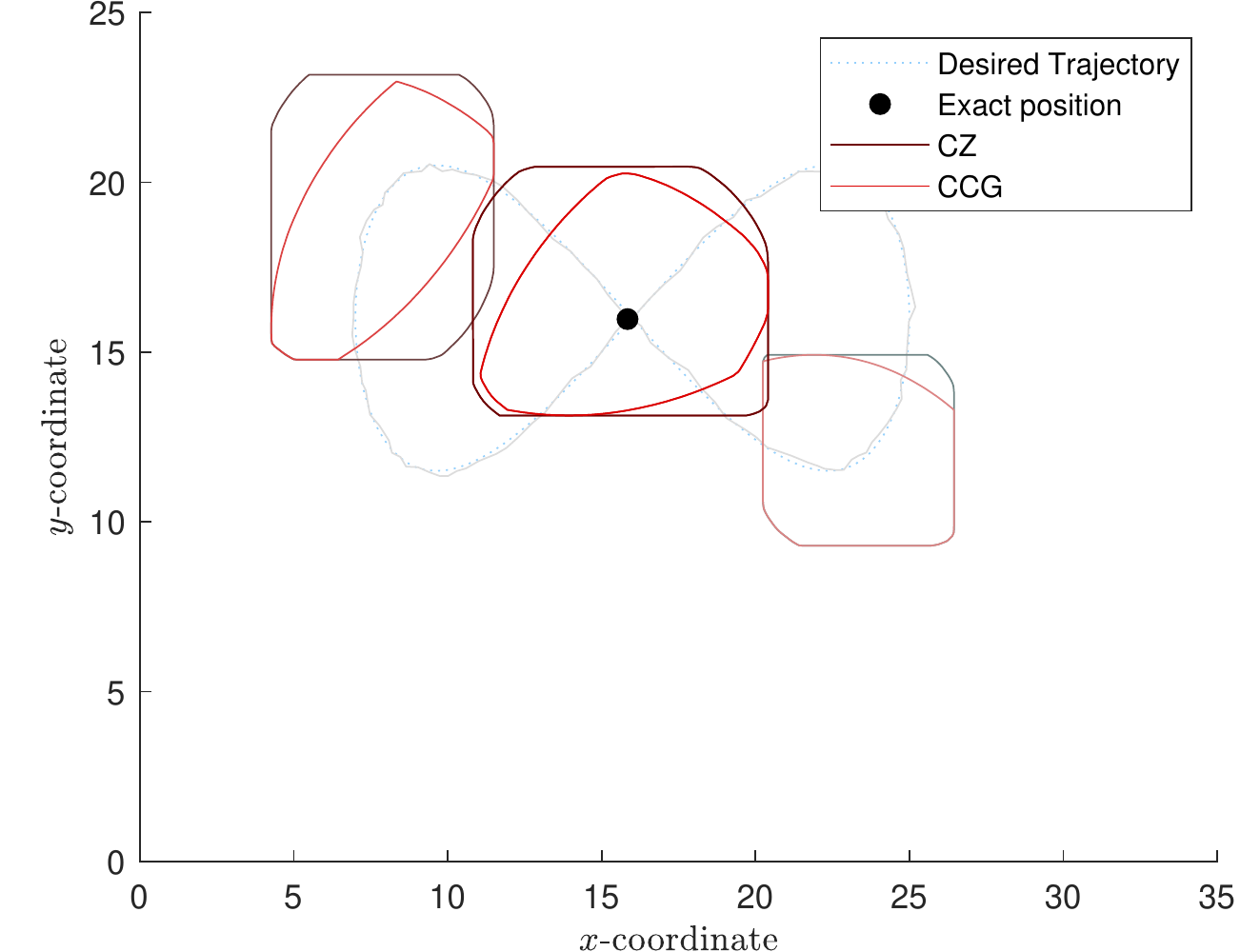}
	\caption{Trajectory executed by the vehicle and the correspondent set-valued estimates at multiples of 40 iterations when using constrained zonotopes (CZ) and CCGs for the figure 8 trajectory.}
	\label{fig:fig2-F8}
\end{figure}

A last relevant issue is the elapsed time in each iteration taken by both filters with different set representations. Figure \ref{fig:fig3-times} shows the computation times across iterations during the whole simulation. At the beginning, both filters have very similar behavior pointing out to the fact that the CCG is yet to have round facets and the order reduction produces equivalent representations. However, as the simulation progresses the set is intersected with the range measurements. The curved boundaries of the CCGs result in a more complex representation. When the vehicle finds the second beacon and the set is considerably reduced in size, the CZ approach has a better performance given that $X(k)$ has a shape close to an interval, where its accuracy is the worst. This result points out to the need to further develop order reduction methods for CCGs that can exploit the nature of the sets. This is not a trivial task given the requirement of computing an outer-approximation to maintain the guaranteed feature in the estimation using set-membership approaches.

\begin{figure}
	\centering
	\includegraphics[scale=0.47,trim=100 271 120 273,clip=true]{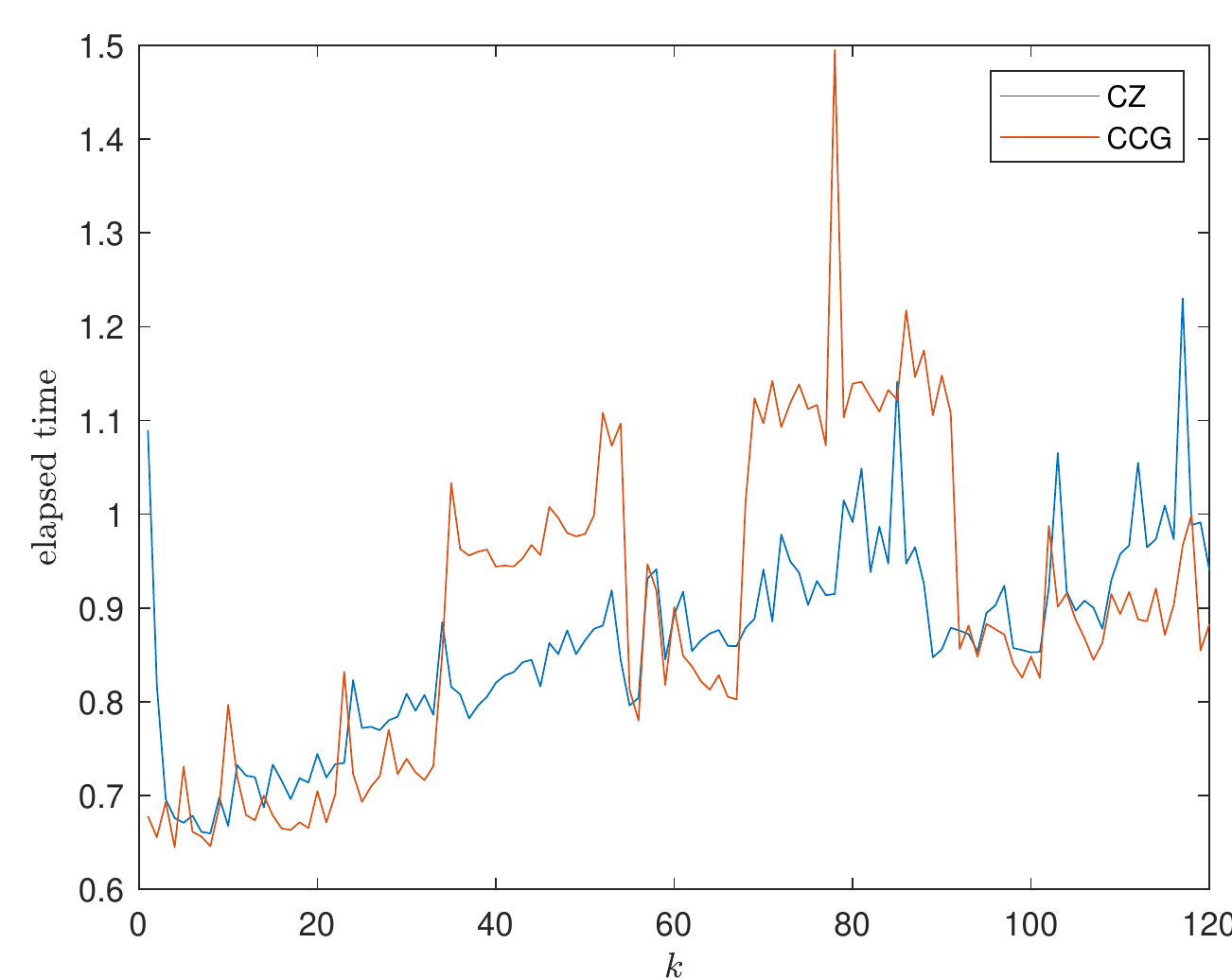}
	\caption{Elapsed time for each iteration of both methods taking into account the constructiong of the set, approximation algorithm and volume computation.}
	\label{fig:fig3-times}
\end{figure}

In order to illustrate an example where both filters should be similar, we simulated a spiral trajectory and increased the range of the beacons in 5 meters each. In this case, the trajectory is not taking advantage of the two beacons. However, the fact that the vehicle will receive the beacon more often should compensate. Figure \ref{fig:fig1-volumes-S} showcases that the volume is indeed much smaller for this trajectory since the vehicle performs dead reckoning less often. In this setup, the main difference between the two filters is precisely the representation of the circular shapes that benefits the CCGs.

\begin{figure}
	\centering
	\includegraphics[scale=0.47,trim=100 271 120 274,clip=true]{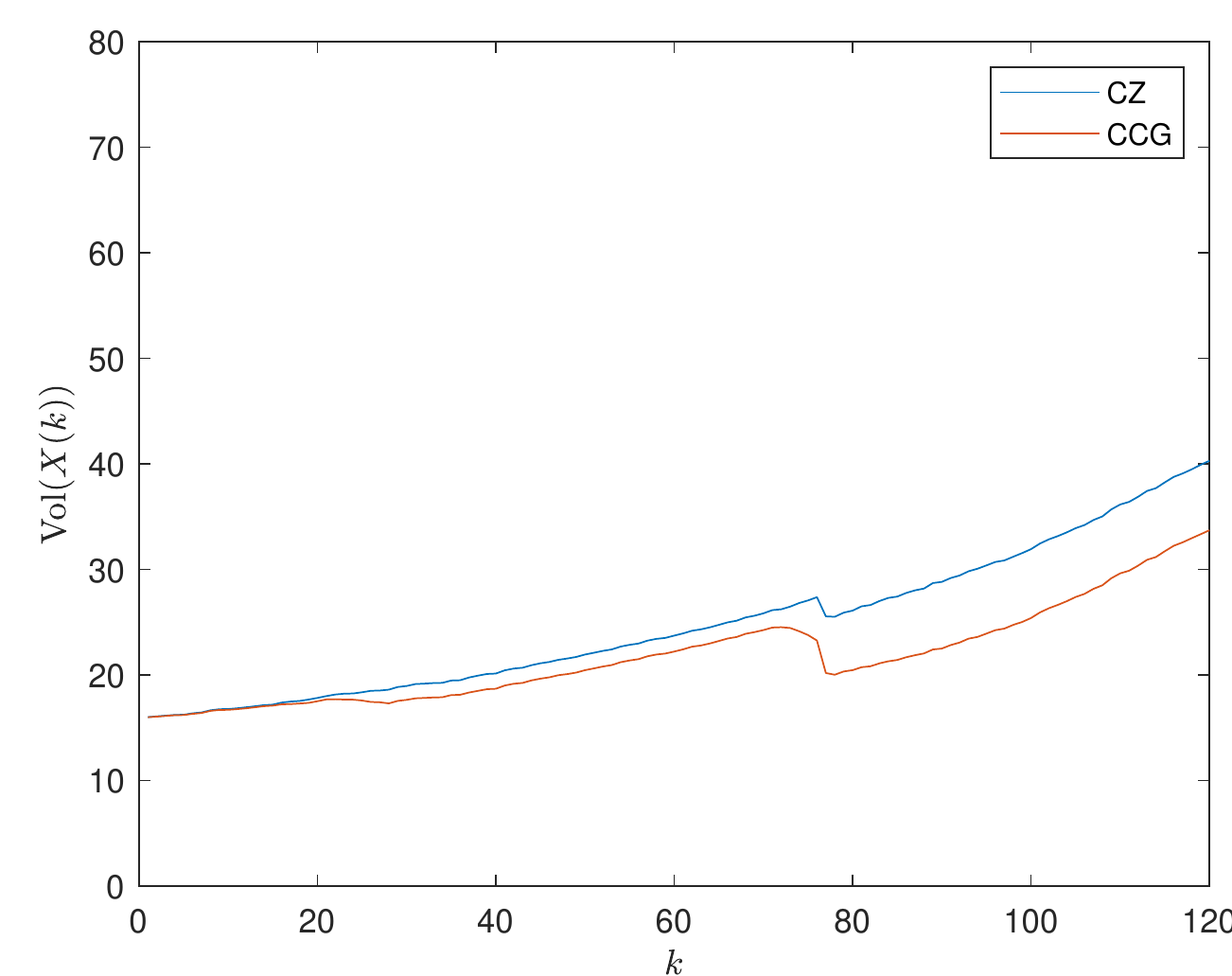}
	\caption{Comparison of the volume for both set-valued estimates when using constrained zonotopes (CZ) and CCGs for the spiral trajectory.}
	\label{fig:fig1-volumes-S}
\end{figure}

In Figure \ref{fig:fig2-S8}, it is depicted the same snapshots for the trajectory where it is noticeable the rounded shapes corresponding to the range measurements. However, as seen in Figure \ref{fig:fig3-times-S}, the more complicated set representation also reduces the performance of both filters. Similarly to the figure 8 trajectory scenario, both simulations illustrate a clear reduction in the conservatism without a very expressive increase in elapsed time for the overall computations. We remark that in terms of orders of magnitude, both filters in normal operation will take between 0.6 and 1.5 seconds, which is not viable for real-time applications and showcases the need to further purse efficient order reductions methods. We did not use the methods from CORA toolbox since we were obtaining even larger computing times. 

\begin{figure}
	\centering
	\includegraphics[scale=0.47,trim=100 271 120 283,clip=true]{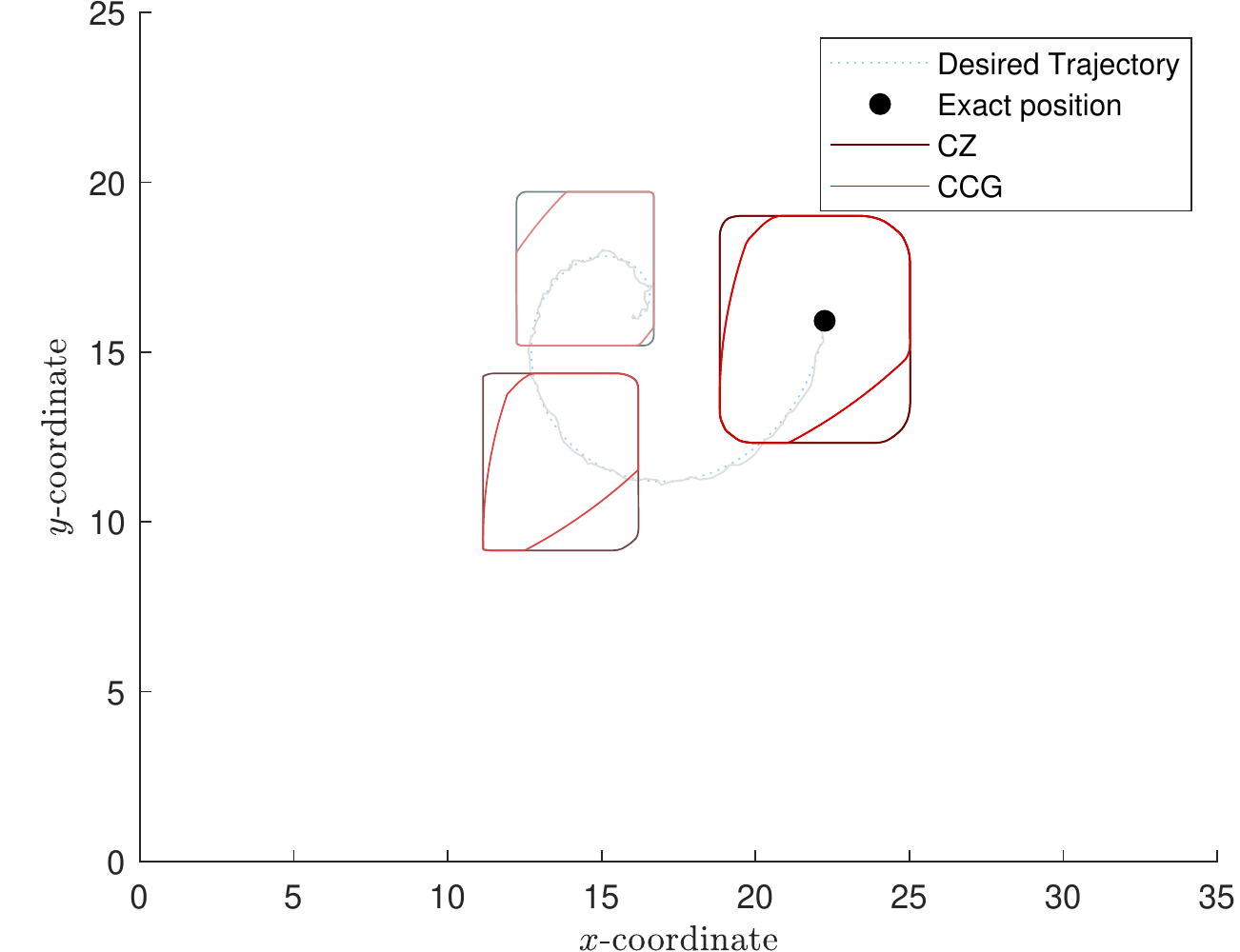}
	\caption{Trajectory executed by the vehicle and the correspondent set-valued estimates at multiples of 40 iterations when using constrained zonotopes (CZ) and CCGs for the spiral trajectory.}
	\label{fig:fig2-S8}
\end{figure}

\begin{figure}
	\centering
	\includegraphics[scale=0.47,trim=100 271 120 273,clip=true]{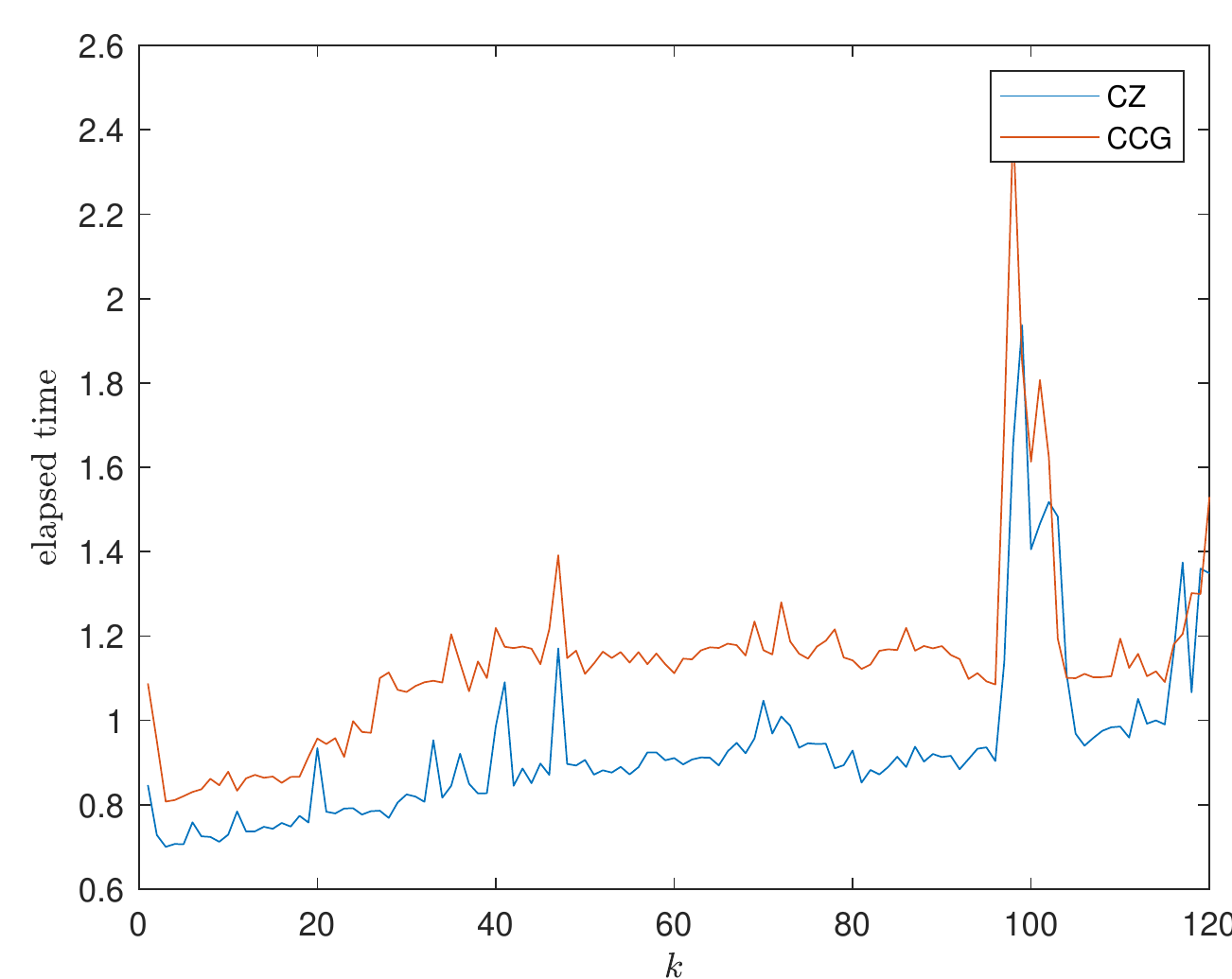}
	\caption{Elapsed time for each iteration of both methods taking into account the constructiong of the set, approximation algorithm and volume computation in the spiral trajectory scenario.}
	\label{fig:fig3-times-S}
\end{figure}

\section{Conclusions and Future Work} \label{sec:discussion}
In this paper, we have address the problem of set-valued estimation for uncertain Linear Parameter-Varying (LPV) models. Given the need for a convex hull operation for the polytopic vertices of the uncertainty space, we develop a closed-form expression tailored for Constrained Convex Generators (CCGs) that is optimal in terms of the number of generators and constraints since it combines both linearly. Given that CZs are a particular instance of CCGs, this results also improves on the state-of-the-art for those methods. 

In a simulation representing a vehicle performing dead reckoning with occasional access to range measurements from beacons, it is shown that the current proposal significantly improves the estimation quality in comparison with CZs that can hardly improve the set-valued estimates. As future work, increasing the performance of order reduction methods for CCGs that take into account the round nature of some of its facets can greatly improve the performance of the filter.  


\bibliography{cvxHull}             
                                                   







\end{document}

%% file: danielcommands.tex

\newcommand{\tp}{^{\intercal}}

\newcommand{\conv}{*}

\newcommand{\T}{\mathcal{T}}

\newcommand{\Prob}{\mathbb{P}}

\newcommand{\R}{\mathbb{R}}

\newcommand{\E}{\mathbb{E}}

\newcommand{\1}{\mathsf{1}}

\newcommand{\0}{\mathsf{0}}

\newcommand{\herm}{^{*}}

\newcommand{\adj}{'}

\newcommand{\sss}[1]
 {\scriptscriptstyle #1}

\newcommand{\vect}{\nu} 

\newcommand{\kron}{\otimes} 
 
\newcommand{\cgeq}{\succeq}  

\newcommand{\cleq}{\preceq}

\newcommand{\cge}{\succ}  

\newcommand{\cle}{\prec}

\newcommand{\inpr}[2]
 {\langle #1,#2 \rangle }

\newcommand{\bs}{\boldsymbol}
\newcommand{\br}[1]{\bs{\mathrm{#1}}}